\documentclass[11pt]{article}

\usepackage{amsmath,amssymb,amsthm}
\usepackage{color}
\usepackage{graphicx}
\usepackage{url}

\newtheorem{theorem}{Theorem}

\newtheorem{lemma}{Lemma}

\newtheorem{proposition}{Proposition}
\newtheorem{remark}{Remark}

\def \beq{ \begin{equation}}
\def \eeq{\end{equation}}

\date{}

\title{On the Nonexistence of Centered Co-Circular Central Configurations  With Three Unequal masses}
\begin{document}

	\maketitle
	\markboth{Zhengyang Tang, Shuqiang Zhu}{  }
	\vspace{-0.5cm}
	\author       
	\bigskip
	\begin{center}
		{Zhengyang Tang$^1$, Shuqiang Zhu$^2$}\\
		{\footnotesize 
			
			School of  Mathematics,  Southwestern University of Finance and Economics, \\
			
			Chengdu 611130, China\\
			
			$^1$\texttt{kevintang\_2003@126.com},  $^2$\texttt{zhusq@swufe.edu.cn}
			
		}
		
	\end{center}

	\begin{abstract}
			This paper examines the existence of centered co-circular central configurations in the general power-law potential \( n \)-body problem. We prove the nonexistence of such configurations when the system consists of \( n-3 \) equal masses and three arbitrary masses, under the condition that the three special masses are distinct or, if two of them are equal, not arranged in a specific manner. 
	\end{abstract}

	
	\textbf{2020AMS Subject Classification}: {  70F10,   70F15}.

	\section{Introductions}

	Central configurations are crucial solutions in the Newtonian \( n \)-body problem. They naturally arise when searching for self-similar solutions and play a key role in identifying bifurcations in the topology of integral manifolds \cite{Smale1970-2, Saari1980, Hampton2006, Albouy2012-1, Albouy2012-2}. These configurations also have important applications in astronomy and spacecraft mission design \cite{Marsden2006}.
	
	In this work, we focus on central configurations where all bodies lie on a common circle and the center of mass coincides with the center of the circle. These are referred to as \emph{centered co-circular central configurations}.
	
	Any regular polygon with equal masses forms a centered co-circular central configuration \cite{Perko1985, Wang2019}. When investigating the question: \emph{Do planar choreography solutions exist when the masses are not all equal?}, Chenciner posed a related problem \cite{Chenciner2004}: \emph{Is the regular polygon with equal masses the unique centered co-circular central configuration?}  This question was later included in a well-known list of open problems on the classical \( n \)-body problem, compiled by Albouy, Cabral, and Santos \cite{Albouy2012-1}. Hampton confirmed the uniqueness for \( n = 4 \) \cite{Hampton2005}, and Wang’s recent research extended the result to \( n = 5 \) and \( n = 6 \) \cite{Wang2023}.
	
	This question is also considered in the general power-law potential \( n \)-body problem, where the potential is given by$
	U_\alpha = \sum_{i< j} \frac{m_i m_j}{|q_i - q_j|^\alpha}$. 
	When \( \alpha = 1 \), the problem corresponds to the Newtonian \( n \)-body problem, while the limiting case \( \alpha = 0 \) corresponds to the \( n \)-vortex problem. For the limiting case \( \alpha = 0 \), Cors, Hall, and Roberts confirmed Chenciner’s conjecture for any \( n \) \cite{Roberts2014}. For \( \alpha > 0 \), Wang provided a positive answer for \( n = 3 \) and \( n = 4 \) \cite{Wang2023}.
	
	Another interesting approach to Chenciner’s question was initiated by Hampton \cite{Hampton2016}, where he demonstrated that no centered co-circular central configuration exists if the masses are sufficiently close to the equal mass case, although the exact size of this neighborhood remains unknown. Subsequent efforts have sought to determine this size. Corbera and Valls \cite{Corbera2019} showed the nonexistence of centered co-circular central configurations for \( n \) equal masses plus one arbitrary mass. Our previous result \cite{Tang2023} extended this by proving nonexistence for \( n \) equal masses plus two arbitrary masses,  for two groups of equal masses,  and for a group of equal masses combined with a group of heavier (lighter) masses, provided the two groups are arranged in an interlaced manner.

	The aim of this paper is to establish the nonexistence of centered co-circular central configurations for \( n \) equal masses and three arbitrary masses. This nonexistence result holds under the condition that the three special masses are distinct or, 
	if two are equal, they are not arranged in a specific way (see Theorem \ref{thm:threeunequal}).

	After the introduction, the paper is organized as follows. In Section \ref{sec:basic}, we briefly present the definitions and necessary lemmas. In Section \ref{sec:main}, we state and prove our main result.

	\section{Basic settings and lemmas}\label{sec:basic}
	In this section, we briefly introduction the central 
	configurations and list some necessary results. 
	
 Let $q_i \in \mathbb{R}^2$ and $m_i$ be the position 
 and mass, respectively, of the $i$-th body. Let  $r_{i j}=\left\|q_i-q_j\right\|$ be
  the distance between the $i$-th and $j$-th bodies. 
The center of mass is given by $c=\frac{\sum_{i=1}^n m_i q_i}{\sum_{i=1}^n m_i}$. 

Consider the power-law potentials $U_\alpha$ of the form
$$
U_\alpha=\sum_{i<j} \frac{m_i m_j}{r_{i j}^\alpha}, \alpha>0.  
$$
A planar \emph{central configuration } is a special set of distinct positions $q_i \in \mathbb{R}^2$ 
satisfying
$$
\sum_{j \neq i}^n \frac{m_i m_j\left(q_j-q_i\right)}{r_{i j}^{\alpha+2}}+\frac{\lambda}{\alpha} m_i q_i=0 \text { for each } i \in\{1, \ldots, n\}
$$
and for some scalar $\lambda$ independent of $i$.
 Summing up all the equations in the above system yields
  $\sum_i m_i q_i=0$, i.e., the center of mass of the 
  central configuration is at the origin.

The term \emph{centered co-circular central configurations} refers to 
  central configurations that  all  bodies lie on a common
   circle and with 
center of mass coincides with the center of 
the circle. Let ${\bf{m}}=\left(m_1, m_2, \ldots, m_n\right)$. 
Without lose of generality, 
let $q_i=\left(\cos \theta_i, \sin \theta_i\right)$, and 
	\[0<\theta_1<\theta_2<\ldots<\theta_n \leq 2 \pi.  \]
	In this way, the  mass vector   determines the order of the masses on the circle. 
	We also write the  configuration by  $\theta =(\theta_1, \ldots, \theta_n).$
Note that $r_{j k}=\left|2 \sin \frac{\theta_j-\theta_k}{2}\right|.$
Obviously, the set of centered central configurations are invariant under rotations of the circle. 
To remove the symmetry, we specify that $\theta_n=2\pi$. Let 
\begin{align*}
	&\mathcal{K}_{0}=\left\{ \mathbf{\theta}: 0<\theta_{1}<\theta_{2}<\ldots<\theta_{n}=2\pi\right\}, \\
&\mathcal{CC}_{0}=\left\{ (\mathbf{\bf m},\theta)\  \mbox{ is a  centered co-circular central configuration},\, \theta\in\mathcal{K}_{0}\right\} .
\end{align*}

		Cors et al. found  (\cite{Roberts2014}) that a configuration is a centered  co-circular central configuration if 
		\begin{equation}\notag
			\frac{\partial}{\partial\theta_{k}}U_{\alpha}=0,\ \ 
			\sum_{j\ne k}\frac{m_{j}}{r_{jk}^\alpha}=\frac{2\lambda}{\alpha},\ k=1,\ldots ,n.
			\label{eq:cce2}
		\end{equation}

	\begin{lemma} [\cite{Roberts2014}]
		\label{lem:Cors}For any $\mathbf{\bf m},$ there is a unique point
		in $\mathcal{K}_{0}$ satisfying $\frac{\partial}{\partial\theta_{k}}U_{\alpha}=0,k=1,\ldots,n.$
		Moreover, the critical point is a minimum, denoted by $\theta_{\bf m}$. 
	\end{lemma}
	
	
	The dihedral group, $D_{n}$, acts on the set $\mathbb{R}_{+}^{n}\times\mathcal{K}_{0}$
	as followes. Denote
	
	\[
	P=\left(\begin{array}{cccccc}
		0 & 1 & 0 & \ldots & 0 & 0\\
		0 & 0 & 1 & \ldots & 0 & 0\\
		. & . & . & \ldots & . & .\\
		0 & 0 & 0 & \ldots & 0 & 1\\
		1 & 0 & 0 & \ldots & 0 & 0
	\end{array}\right),\ \ S=\left(\begin{array}{cccccc}
		0 & 0 & \ldots & 0 & 1 & 0\\
		0 & 0 & \ldots & 1 & 0 & 0\\
		. & . & \ldots & . & . & .\\
		1 & 0 & \ldots & 0 & 0 & 0\\
		0 & 0 & \ldots & 0 & 0 & 1
	\end{array}\right),
	\]
	\[
	\mathcal{P}=\left(\begin{array}{cccccc}
		-1 & 1 & 0 & \ldots & 0 & 0\\
		-1 & 0 & 1 & \ldots & 0 & 0\\
		. & . & . & \ldots & . & .\\
		-1 & 0 & 0 & \ldots & 0 & 1\\
		0 & 0 & 0 & \ldots & 0 & 1
	\end{array}\right),\ \ \mathcal{S}=\left(\begin{array}{cccccc}
		0 & 0 & \ldots & 0 & -1 & 1\\
		0 & 0 & \ldots & -1 & 0 & 1\\
		. & . & \ldots & . & . & .\\
		-1 & 0 & \ldots & 0 & 0 & 1\\
		0 & 0 & \ldots & 0 & 0 & 1
	\end{array}\right).
	\]
	The action of $D_{n}$ on $\mathbb{R}_{+}^{n}$ is by the matrix group
	generated by $P,S$, and the action of $D_{n}$ on $\mathcal{K}_{0}$
	is by the matrix group generated by $\mathcal{P},\mathcal{S}$. For
	any $g=P^{h}S^{l}\in D_{n},$ letting $\hat{g}=\mathcal{P}^{h}\mathcal{S}^{l},$define
	the action of $D_{n}$ on $\mathbb{R}_{+}^{n}\times\mathcal{K}_{0}$
	by 
	\[
	g\cdot(\mathbf{m,}\theta)=(g\mathbf{\bf m},\hat{g}\theta).
	\]

	\begin{lemma}\cite{Tang2023} \label{lem:key0}
		Assume that $(\mathbf{m,}\theta_{\mathbf{\bf m}})\in\mathcal{CC}_{0}$
		is a centered co-circular central configuration, then
		\begin{enumerate}
			\item For any $g\in D_{n}$, $g\cdot(\mathbf{m,}\theta_{\mathbf{\bf m}})\in\mathcal{CC}_{0}$. 
			\item $U_{\alpha}(\mathbf{m,}\theta_{\mathbf{\bf m}})=U_{\alpha}(g\mathbf{m,}\hat{g}\theta_{\mathbf{\bf m}})\le U_{\alpha}(g\mathbf{m,}\theta_{\mathbf{\bf m}})$
			and $\hat{g}\theta_{\mathbf{\bf m}}=\theta_{g\mathbf{\bf m}}.$ 
			\item $\mathbf{m=}g\mathbf{\bf m}$ implies $\hat{g}\theta_{\mathbf{\bf m}}=\theta_{\mathbf{\bf m}}$. 
		\end{enumerate}
	\end{lemma}
	
	Consider the symmetric matrix $H_{\mathbf{{\bf m}}}$, which is determined
	by $\mathbf{m}$ and the corresponding $\theta_{{\bf m}}$, by $\text{(\ensuremath{H_{{\bf m}})_{ij}}=1/\ensuremath{r_{ij}^{\alpha}}}$ when
	$i\ne j$, and $(\ensuremath{H_{{\bf m}})_{ii}}=0.$  A direct
	consequence of Lemma \ref{lem:key0} is the following. 

	\begin{lemma}\cite{Tang2023}
		\label{lemma:key1} Given $\mathbf{m}$ and the corresponding $\theta_{\bf m}$,
		if there is some $g\in D_{n}$ such that $H_{\mathbf{\bf m}}(g \mathbf{m}-\mathbf{m})<0$,
		then $(\mathbf{m},\theta_{\bf m})\notin\mathcal{CC}_{0}$. 
	\end{lemma}

	When investigating the sign of $H_{\mathbf{\bf m}}(g\text{\textbf{m}}-\mathbf{m})$, 
	the following fact  is useful. 
	\begin{lemma} \cite{Tang2023}
		\label{lemma:keyquadri} Suppose that four points
		A, B, C, D, are ordered counterclockwise on a circle. See Figure  \ref{fig:2quadri}.  Then for
		any $\alpha>0,$ it holds that
		
		\[
		\frac{1}{AC^{\alpha}}+\frac{1}{BD^{\alpha}}-(\frac{1}{AD^{\alpha}}+\frac{1}{BC^{\alpha}})<0.
		\]
	\end{lemma}
	
	\begin{figure}[h!]
		\centering
		\includegraphics[scale=0.6]{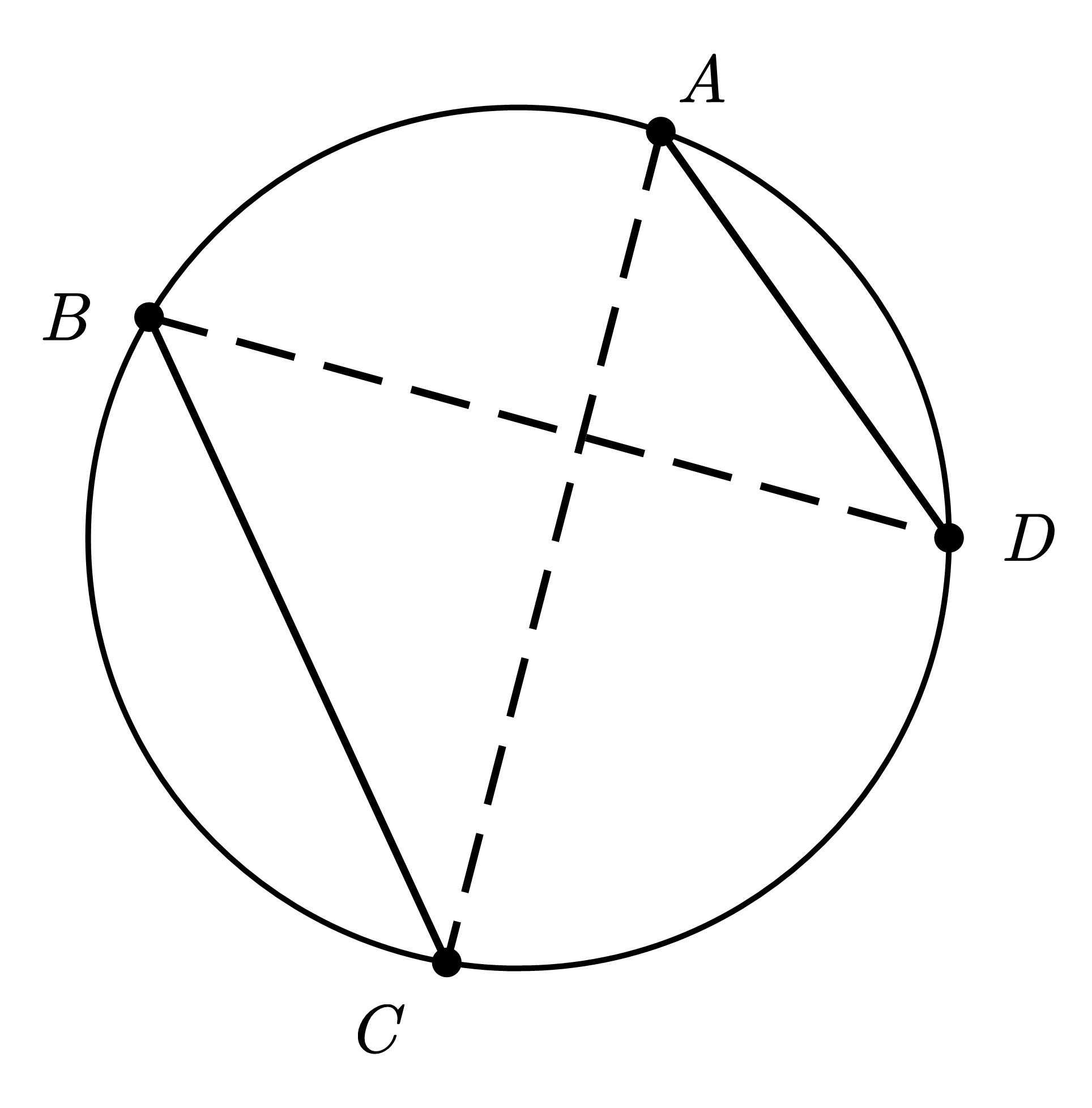}  
		\caption{The dashed lines correspond to the positive terms, while the solid black lines  correspond to the negative terms.
		}
		\label{fig:2quadri}
	\end{figure}
	
The following two results will be essential for our discussion:

		\begin{theorem} \cite{Tang2023}
		\label{thm:twounequal} In the general power
		law potential n-body problem, no centered co-circular central configurations exist when precisely $n-2$ of the masses are equal.
	\end{theorem}

		\begin{theorem} \cite{Tang2023}
		\label{thm:2group-of-equal-masses} 
		In the general power-law potential n-body problem,  when masses can be grouped into two sets of equal masses, no centered co-circular central configurations  exist  unless all masses are equal.
	\end{theorem}

	\section{Nonexistence of centered co-circular central configurations }\label{sec:main}

In this section, we will establish the following result:
	\begin{theorem}
		\label{thm:threeunequal} Consider a co-circular configuration $(\mathbf{m},\theta_{\bf m})$ of 
		the general power-law potential $n$-body problem. Assume all but three masses are equal and that the three special masses are not ordered symmetrically.  Then this configuration cannot be a centered co-circular central configuration.
	\end{theorem}
	
The term \emph{special mass} refers to one of the three unequal masses. We say that three masses are \emph{ordered symmetrically} in a co-circular configuration if two of the masses are equal, and the number of masses between each of the two equal masses and the third mass (counted along the circle) is the same.   See Figure \ref{fig:symmetric}. 
	
	 \begin{figure}[h!]
		\centering
	\includegraphics[scale=0.6]{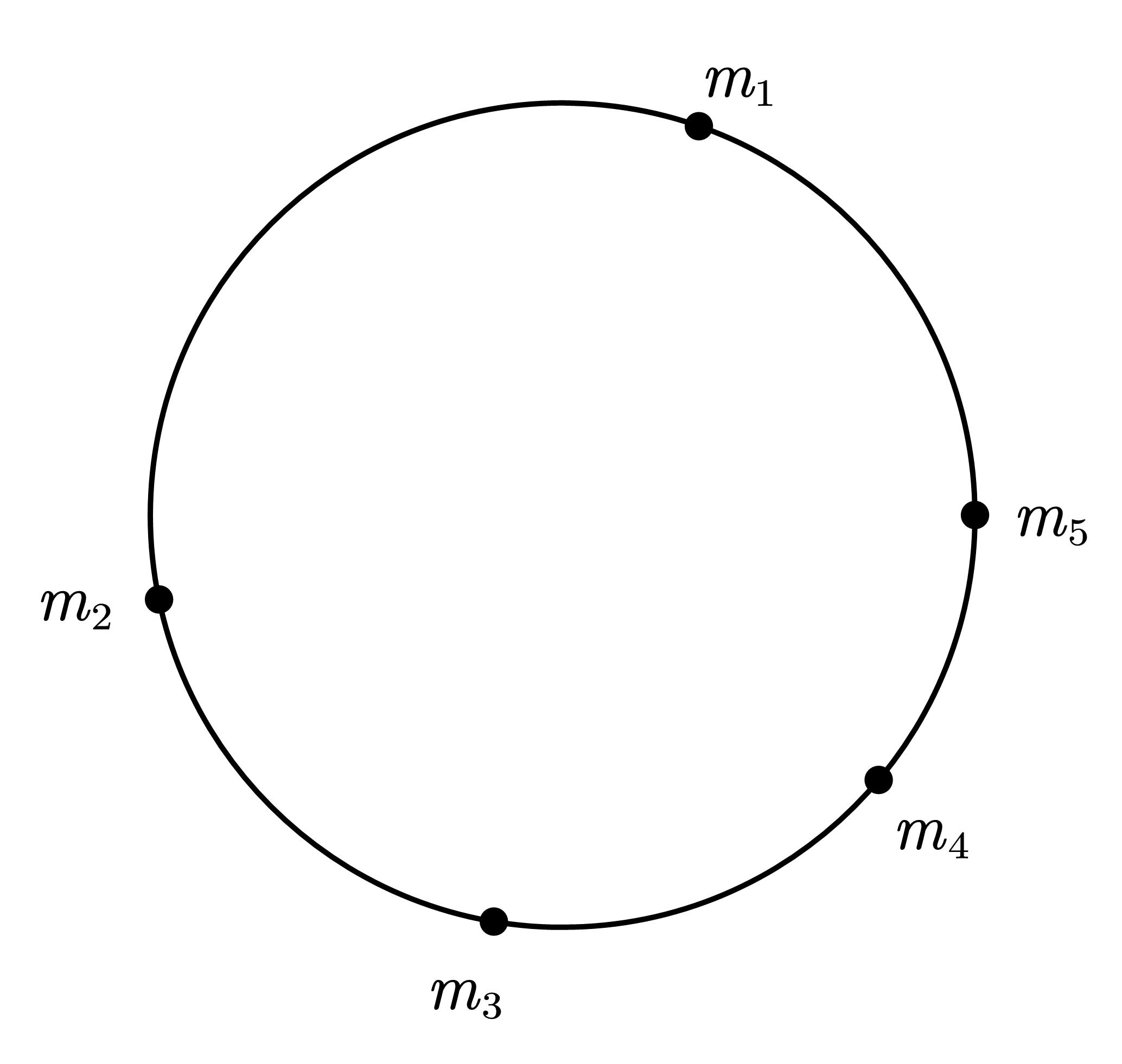} 
		\caption{A five-body co-circular configuration with $m_2=m_3$.   
		We say that the three masses $m_2, m_3, m_5$ are ordered symmetrically.   }
		\label{fig:symmetric}
	\end{figure}

	The idea of the proof is the following:  
	 Consider the corresponding  
	quadratic form $H_{\bf m}$. If there is a   $g\in D_n$  so that $H_{\bf m}(g\text{\textbf{m}}-\mathbf{m})$ is negative, then Lemma \ref{lemma:key1} yields the desired result.

		Without lose of generality, assume that the equal masses are 1. Assume 
		the three special masses are  positioned  at the $l$-th, $s$-th, and $n$-th locations ($1\le l<s<n$) respectively,  their masses are $m_l, m_s, m_n$, and   these masses are not equal to $1$. Thus, the mass vector is: 
		$$
		\mathbf{m}=(1,\ldots,1,m_l,1,\ldots,1,m_s,1,\ldots,1,m_n)
		$$

	We start by considering the following specific case:
	\begin{proposition}
		\label{prop:threesp}
		Consider a co-circular configuration $(\mathbf{m},\theta_{\bf m})$ of the general power-law potential \( n\)-body problem. Suppose that  $n$ is even,  all but three masses are equal,  and two  special masses 
		are separated by exactly \( \frac{n}{2} -1\) masses (counted on the circle). Then 
		the configuration is not a  centered co-circular central configuration.  
	\end{proposition}

	\begin{proof}
		Without loss of generality,  assume that $m_s$ and $m_n$ are   separated by  \( \frac{n}{2} -1\) masses. Then $s= \frac{n}{2}$ and 
		 the mass distribution
		$$
		\mathbf{m}=(1,\ldots,1,m_l,1,\ldots,1,m_\frac{n}{2},1,\ldots,1,m_n)
		$$
		Let $g=S$. Then
		$$
		g\mathbf{m}-\mathbf{m}= (1-m_l)(0,\ldots,1,0,\ldots,0,-1,0,\ldots,0), 
		$$
		where the nonzero terms are  at the $l$-th  and $(n-l)$-th positions. 
		So, 
		$$
		H_{\mathbf{\bf m}}(g\text{\textbf{m}}-\mathbf{m})= -2(1-m_l)^2\frac{1}{r^\alpha_{l, n-l}} <0, 
		$$
		By Lemma \ref{lemma:key1}, $(\mathbf{m},\theta_{\bf m})$ is not 
		a centered co-circular central configuration. 
	\end{proof}
	
    Hence, in the following, we assume that there are not two special masses  separated by  \( \frac{n}{2} -1\) masses, or equivaently, the following holds:
	\begin{equation}\notag \label{equ:no-opposite}
	 s-l- \frac{n}{2}\ne 0, \, n-s-\frac{n}{2}\ne 0,\,  n-l- \frac{n}{2}\ne 0.   
	\end{equation}

	We now establish two critical lemmas to handle the general cases.
	
	\begin{lemma}\label{lemma:one-side}
		Suppose that $(\mathbf{m},\theta_{\bf m})$ is a centered co-circular central configuration 
		of the general power-law potential $n$-body problem. 
		Assume that  all masses equal except three. 
		If  $(m_{l}-1)(m_{s}-1)<0$, then 
		$(s-\frac{n}{2})(l-\frac{n}{2})<0$.  
	\end{lemma}

	\begin{figure}[h!]
		\centering
	  \includegraphics[scale=0.6]{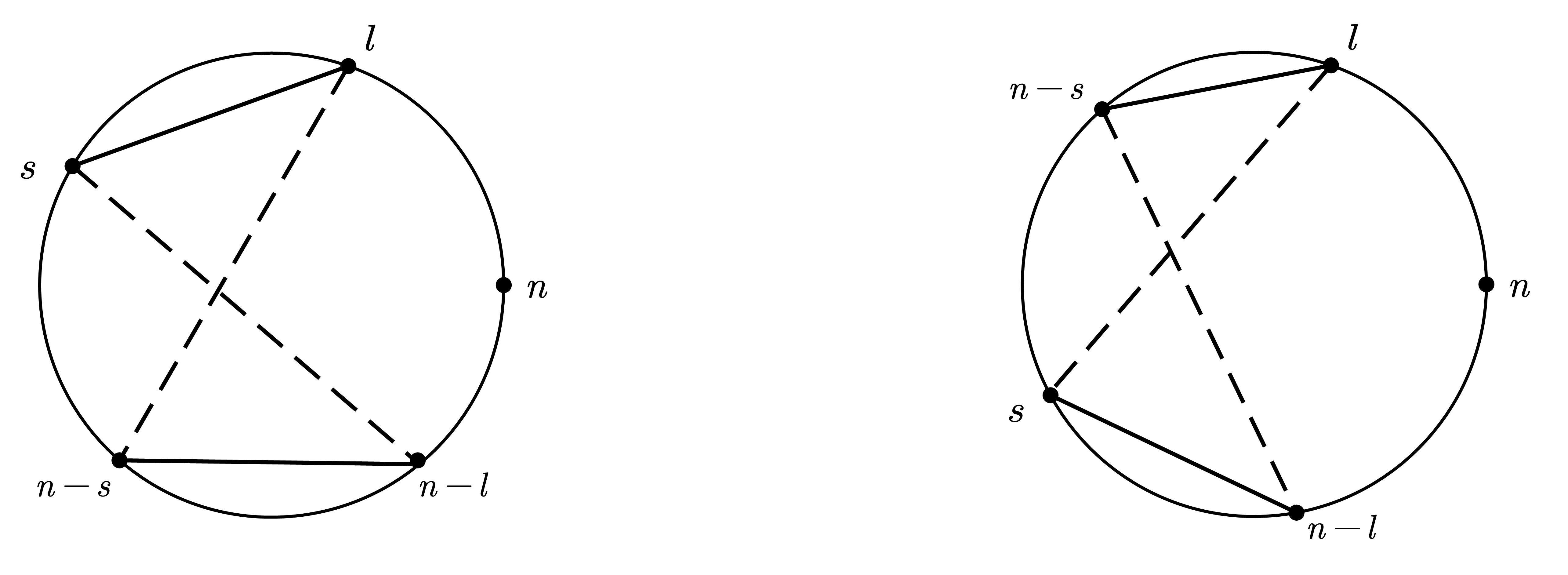}  
		\caption{Left: \( m_l \) and \( m_s \)  both lie in the first half of the 
		circle (i.e., among the positions \( 1, 2, ..., [\frac{n}{2}] \) );  ($[c]$ denotes the greatest integer less than or equal to $c$. )
		Right: \( m_l \)  lies in the first half of the circle, \( m_s\)  lies in the second half, and $\frac{n}{2} -l > s-\frac{n}{2}$.  The dashed lines correspond to the positive terms, while the solid  lines  correspond to the negative terms.
		}
		\label{fig:sameside}
	\end{figure}

	\begin{proof}
Suppose for contradiction  that $(s-\frac{n}{2})(l-\frac{n}{2})\ge 0$.  Since the masses are not separated by $\frac{n}{2}-1$ masses,  we assume  $s<\frac{n}{2}$. See Figure \ref{fig:sameside}, left. 
		The mass distribution is 
		$$
		\mathbf{m}=(1,\ldots,1,m_l,1,\ldots,1,m_s,1,\ldots,1,m_n)
		$$
		Let $g=S$.  Then
		$$
		g\mathbf{m}-\mathbf{m}=(0,\ldots,0,1-m_l,0,\ldots,0,1-m_s,0,\ldots,0,m_s-1,0,\ldots,0,m_l-1,0,\ldots,0), 
		$$
		where the nonzero terms are  at the $l$-th,  $s$-th, $(n-s)$-th, and $(n-l)$-th positions. 
		So, 
		\begin{align*}
		H_{\mathbf{\bf m}}(g\text{\textbf{m}}-\mathbf{m})&=
	-\frac{(1-m_l)^2}{r_{l,n-l}^\alpha}-\frac{(1-m_s)^2}{r_{s,n-s}^\alpha}+\\
&	2(1-m_l)(m_s-1) \left( - \frac{1}{r_{l s}^{\alpha}}+   \frac{1}{r_{l, n-s}^{\alpha}}    +  \frac{1}{r_{s, n-l}^{\alpha}}  - \frac{1}{r_{n-s, n-l}^{\alpha}}    \right)\\
	&<0, 
		\end{align*}
which contradicts with  Lemma \ref{lemma:key1}. 	Here, the inequality is by Lemma \ref{lemma:keyquadri}. 
	Hence, $(s-\frac{n}{2})(l-\frac{n}{2})<0$.   
	\end{proof}

\begin{remark}\label{rmk:one-side}
Under the same assumptions as Lemma \ref{lemma:one-side}, 
we can similarly deduce:
\begin{itemize}
	\item If $(m_{s}-1)(m_{n}-1)<0$, then $(s-l-\frac{n}{2})(n-l-\frac{n}{2})<0$; 
	\item If $(m_{l}-1)(m_{n}-1)<0$, then $(l+n -s-\frac{n}{2})(n-s-\frac{n}{2})<0$. 
\end{itemize}
\end{remark}

	\begin{lemma}\label{lemma:two-sides}
		Suppose that $(\mathbf{m},\theta_{\bf m})$ is a centered co-circular central configuration 
		of the general power-law potential $n$-body problem. 
		Assume that all masses are equal except three and that the three special masses are not  ordered symmetrically. 
		If $(m_{l}-1)(m_{s}-1)>0,$   then  $(l-\frac{n}{2})(s-\frac{n}{2})>0$. 
	\end{lemma}

	\begin{proof}
	Suppose  by contradiction that $(s-\frac{n}{2})(l-\frac{n}{2})<0$.  Then $l<\frac{n}{2}$ and  $s>\frac{n}{2}$.  
		There are two cases: 
		
		1. $l+s=n$. Then by assumption,  it holds that $m_l\ne m_s$.
		Let $g=S$.  Then
		$$
		g\mathbf{m}-\mathbf{m}=(0,\ldots,0,m_s-m_l,0,\ldots,0,m_l-m_s,0,\ldots,0), 
		$$
	where the nonzero terms are  at the $l$-th,  and $s$-th positions. Thus, 
		$$
		H_{\mathbf{m}}(g\mathbf{m}-\mathbf{m})=-(m_l-m_s)^2\frac{1}{r_{l,s}}<0. 
		$$
		
		 2. $l+s\ne n$.	 Without lose of generality, assume $\frac{n}{2} -l > s-\frac{n}{2}$. See Figure \ref{fig:sameside}, right. 
			Let $g=S$.  Then
		$$
		g\mathbf{m}-\mathbf{m}=(0,\ldots,0,1-m_l,0,\ldots,0,m_s-1,0,\ldots,0,1-m_s,0,\ldots,0,m_l-1,0,\ldots,0), 
		$$
	where the nonzero terms are  at the $l$-th,  $(n-s)$-th , $s$-th, and $(n-l)$-th positions. Thus, 
		\begin{align*}
	H_{\mathbf{\bf m}}(g\text{\textbf{m}}-\mathbf{m})&=
	-\frac{(1-m_l)^2}{r_{l,n-l}^\alpha}-\frac{(1-m_s)^2}{r_{s,n-s}^\alpha}+\\
	&	2(m_l-1)(m_s-1) \left( -\frac{1}{r_{l, n-s}^{\alpha}}   + \frac{1}{r_{l s}^{\alpha}}+   \frac{1}{r_{n-s, n-l}^{\alpha}}  -   \frac{1}{r_{s, n-l}^{\alpha}}   \right)\\
	&<0. 
	\end{align*}
	Here, the inequality is by Lemma \ref{lemma:keyquadri}.

	By Lemma \ref{lemma:key1}, for both cases,  there is no centered co-circular central configuration, a contradiction. Hence, $(s-\frac{n}{2})(l-\frac{n}{2})>0$.   
		\end{proof}

\begin{remark}\label{rmk:two-sides}
	Under the same assumptions as Lemma \ref{lemma:two-sides}, 
	we can similarly deduce:
	\begin{itemize}
		\item If $(m_{s}-1)(m_{n}-1)>0$, then $(s-l-\frac{n}{2})(n-l-\frac{n}{2})>0$; 
		\item If $(m_{l}-1)(m_{n}-1)>0$, then $(l+n -s-\frac{n}{2})(n-s-\frac{n}{2})>0$. 
	\end{itemize}
\end{remark}

	Now, we discuss the case when two masses are equal. 

	\begin{proposition}\label{prp:two-equal-masses}
		Consider a co-circular configuration $(\mathbf{m},\theta_{\bf m})$ of 
		the general power-law potential $n$-body problem. Assume that all masses are equal except three and that the three special
		 masses are not  ordered symmetrically.
		 		If two of the special masses are equal,  
		then the configuration is not a  centered co-circular central configuration. 
	\end{proposition}

\begin{proof}
	 Without lose of generality, assume that $m_l =m_s$.  
	  Suppose  by contradiction  that $(\mathbf{m}, \theta_{\mathbf{m}}) \in \mathcal{CC}_0$.  Since $(m_l-1)(m_s-1)>0$, 
	  Lemma \ref{lemma:two-sides} yields  
	  $ (l-\frac{n}{2}) (s-\frac{n}{2})>0.$ 
	  Without lose of generality, we further assume that $1\le l< s<\frac{n}{2}$. See Figure \ref{fig:sameside}, left. 
 
 On one hand, 
   note that 
 \[  (l+n-s-\frac{n}{2}) (n-s-\frac{n}{2})>0.  \] 
  Remark \ref{rmk:one-side} after Lemma \ref{lemma:one-side} implies the inequality 
  	$(m_l-1)(m_n-1)\ge 0$. 
   By Theorem \ref{thm:twounequal}, it can not be true that $(m_l-1)(m_n-1)=0$, so 
   \begin{equation}\label{equ:contradiction} 
(m_l-1)(m_n-1)> 0.   
   \end{equation}

  On the other hand, 
 note that 
 \[  (s-l-\frac{n}{2}) (n-l-\frac{n}{2})<0.  \] 
 Remark \ref{rmk:two-sides} after Lemma \ref{lemma:two-sides} implies the inequality 
$	(m_s-1)(m_n-1)\le  0$. 
 By Theorem \ref{thm:twounequal}, it can not be true that $(m_s-1)(m_n-1)=0$, so 
 \[    	(m_s-1)(m_n-1)< 0.    \]  
 This  
 contradicts with \eqref{equ:contradiction}  since $m_s=m_l$. 

Hence, $(\mathbf{m}, \theta_{\mathbf{m}})$ is not a centered co-circular central configuration.
\end{proof}
	
	Now, we are ready to prove Theorem \ref{thm:threeunequal}. 
	
\begin{proof} [Proof of Theorem \ref{thm:threeunequal}] 
	
By Theorem \ref{thm:2group-of-equal-masses}, we assume that the three special 
masses are not all equal.
Therefore, in the following, withou any lose of generality, 
	we assume further that $m_l< m_s< m_n$. There are three cases:

1. $1<m_l<m_s<m_n$, or, $m_l<m_s<m_n<1$.  We exclude this case by  contradiction.   Suppose that $(\mathbf{m}, \theta_{\mathbf{m}}) \in \mathcal{CC}_0$. 
We claim that  the following system of inequalities holds
\begin{equation*}
	\begin{cases}
		(l-\frac{n}{2})(s-\frac{n}{2})>0,\\
		(s-l-\frac{n}{2})(n-l-\frac{n}{2})>0. 
	\end{cases}
\end{equation*}
The first inequality is from the fact  $(m_l-1) (m_s-1)>0$ and  Lemma \ref{lemma:two-sides}. The second one is from the fact  $(m_s-1) (m_n-1)>0$ and Remark \ref{rmk:two-sides} after  Lemma \ref{lemma:two-sides}. 
Obviously,  no solutions exist because these inequalities are mutually incompatible. 
Hence, we exclude this case. 
		
		2. $m_l<1<m_s<m_n$.  We exclude this case by  contradiction.   Suppose that $(\mathbf{m}, \theta_{\mathbf{m}}) \in \mathcal{CC}_0$. 
		We claim that  the following system of inequalities holds
\begin{equation*}
	\begin{cases}
		(l-\frac{n}{2})(s-\frac{n}{2})<0, \\
		(s-l-\frac{n}{2})(n-l-\frac{n}{2})>0, \\  
		(l+n -s-\frac{n}{2})(n-s-\frac{n}{2})<0. 
	\end{cases}
\end{equation*}
The first inequality is from the fact  $(m_l-1) (m_s-1)<0$ and  Lemma \ref{lemma:one-side}. The second one is from the fact  $(m_s-1) (m_n-1)>0$ and Remark \ref{rmk:two-sides} after  Lemma \ref{lemma:two-sides}. The third one is from the fact  $(m_l-1) (m_n-1)<0$ and Remark \ref{rmk:one-side} after  Lemma \ref{lemma:one-side}. 
Obviously,  no solutions exist because these inequalities are mutually incompatible. 
Hence, we exclude this case.

		3. $m_l<m_s<1<m_n$. We exclude this case by  contradiction.   Suppose that $(\mathbf{m}, \theta_{\mathbf{m}}) \in \mathcal{CC}_0$. 
		We claim that  the following system of inequalities holds
		\begin{equation*}
			\begin{cases}
				(l-\frac{n}{2})(s-\frac{n}{2})>0, \\
	(l+n -s-\frac{n}{2})(n-s-\frac{n}{2})<0. 
			\end{cases}
		\end{equation*}
	The first inequality is from the fact  $(m_l-1) (m_s-1)>0$ and  Lemma \ref{lemma:two-sides}. The second one is from the fact  $(m_l-1) (m_n-1)<0$ and Remark \ref{rmk:one-side} after  Lemma \ref{lemma:one-side}. 
	Obviously,  no solutions exist because these inequalities are mutually incompatible. 
	Hence, we exclude this case.

		Therefore,  the proof is completed. 
	\end{proof}

\end{document}